\documentclass[10pt]{amsart}

\usepackage{amssymb,amsthm,amsmath}
\usepackage[numbers,sort&compress]{natbib}
\usepackage{color}
\usepackage{graphicx}
\usepackage{tikz}

\title[ ]{ANDERSON LOCALIZATION FOR THE QUANTUM KICKED ROTOR MODEL }

\author{Jia Shi}
\address[Jia Shi]
{School of Mathematical Sciences,
Fudan University,
Shanghai 200433, China} \email{15110180007@fudan.edu.cn}

\author{Xiaoping Yuan}
\address[Xiaoping Yuan]
{School of Mathematical Sciences,
Fudan University,
Shanghai 200433, China} \email{xpyuan@fudan.edu.cn}


\theoremstyle{plain}
\newtheorem{thm}{Theorem}[section]
 
 \newtheorem{lem}[thm]{Lemma}
 \newtheorem{prop}[thm]{Proposition}
 \newtheorem{rem}[thm]{Remark}
 
 \numberwithin{equation}{section}

\begin{document}


\begin{abstract}
In this paper, we establish Anderson localization for the quantum kicked rotor model.
More precisely, we proved that

\begin{equation*}
H=\tan\pi\left(x_0+my_0+\frac{m(m-1)}{2}\omega\right) \delta_{mn}+\epsilon S_\phi
\end{equation*}
has pure point spectrum with exponentially decaying eigenfunctions
for almost all $\omega \in DC$ (diophantine condition).

\end{abstract}
\maketitle
\section{Introduction and main result}

Anderson localization for quasi-periodic Schr\"odinger operators is an important topic in both physics and mathematics.
For example, we can study
\begin{equation}\label{1}
H=v_n\delta_{nn'}+\Delta,
\end{equation}
where $v_n$ is a quasi-periodic potential and $\Delta$ is the lattice Laplacian on $\mathbb{Z}$
\begin{equation*}
\Delta(n,n')=1, |n-n'|=1,\quad \Delta(n,n')=0, |n-n'|\neq 1.
\end{equation*}

Anderson localization means that $H$ has pure point spectrum with exponentially decaying eigenfunctions.
Since there are many papers on this topic, we only mention some results here.
For more about dynamics and spectral theory of quasi-periodic Schr\"odinger-type operators, see the survey \cite{MJ}.

We may associate the potential $v_n$ to a dynamical system $T$ as follows:
\begin{equation}\label{2}
v_n=\lambda v(T^nx),
\end{equation}
where $v$ is a nonconstant real analytic potential on $\mathbb{T}$ and $T$ is a shift
\begin{equation}\label{3}
Tx=x+\omega.
\end{equation}
Fix $x=x_0$, Bourgain and Goldstein \cite{BG} proved that if $\lambda >\lambda_0$, for almost all $\omega$,
$H$ will satisfy Anderson localization. Their argument is based on a combination of large deviation estimates and general facts on semi-algebraic sets.
The method in \cite{BG} depends explicitly on the fundamental matrix and Lyapounov exponent.
By multi-scale method, Bourgain, Goldstein and Schlag \cite{BGS2}  proved Anderson localization for Schr\"odinger operators on $\mathbb{Z}^{2}$
\begin{equation*}
H(\omega_1,\omega_2;\theta_1,\theta_2)=\lambda v(\theta_1+n_1\omega_1,\theta_2+n_2\omega_2)+\Delta.
\end{equation*}
Later, Bourgain \cite{B07} proved Anderson localization for quasi-periodic lattice Schr\"odinger operators on $\mathbb{Z}^{d}$, $d$ arbitrary.
Recently, using more elaborate semi-algebraic arguments, Bourgain and Kachkovskiy \cite{BK}
proved Anderson localization for two interacting quasi-periodic particles.

More generally, we can consider the long range model
\begin{equation}\label{4}
H=v(x+n\omega)\delta_{nn'}+\epsilon S_\phi,
\end{equation}
where $S_\phi$ is a Toeplitz operator
\begin{equation*}
S_\phi(n,n')=\hat{\phi}(n-n')
\end{equation*}
and $v$ is real analytic, nonconstant on $\mathbb{T}$.
Assume $\phi$  real analytic satisfying
\begin{equation}\label{5}
  |\hat{\phi}(n)|<e^{-\rho|n|}, \quad\forall n \in \mathbb{Z}
\end{equation}
for some $\rho>0$, Bourgain \cite{B05} proved that there is $\epsilon_{0}=\epsilon_{0}(\rho)>0$, such that if
$0<\epsilon<\epsilon_{0}$, $H$ satisfies Anderson localization.
Note that in the long range case, we cannot use the fundamental matrix formalism.
The method in \cite{B05} can also be used to establish Anderson localization for band Schr\"odinger operators \cite{BJ}
\begin{equation*}
H_{(n,s),(n',s')}(\omega,\theta)=\lambda v_s(\theta+n\omega)\delta_{nn'}\delta_{ss'}+\Delta,
\end{equation*}
where $\{v_s|1\leq s\leq b\}$ are real analytic, nonconstant on $\mathbb{T}$.
Recently, this method was used to prove Anderson localization for
the long-range quasi-periodic block operators \cite{JSY}
\begin{equation*}
(H(x)\vec{\psi})_n=\epsilon\sum_{k\in\mathbb{Z}}W_k\vec{\psi}_{n-k}+V(x+n\omega)\vec{\psi}_n,
\end{equation*}
where
$$V(x)=\mathrm{diag} (v_1(x),\ldots,v_l(x)),$$
$v_i(x)\  (1\leq i\leq l)$ are nonconstant real analytic functions on $\mathbb{T}$ and $W_k\  (k\in\mathbb{Z})$
are $l\times l$ matrices satisfying $W_k^*=W_{-k}$, $\|W_k\|\leq e^{-\rho|k|},\rho>0$.
For the Anderson localization results of long-range quasi-periodic operators on $\mathbb{Z}^{d}$, we refer to \cite{CD}, \cite{JLS}.

Now, let $T$ be the skew shift on $\mathbb{T}^{2}$:
\begin{equation}\label{6}
T(x_{1},x_{2})=(x_{1}+x_{2},x_{2}+\omega),
\end{equation}
using transfer matrix and Lyapounov exponent, Bourgain, Goldstein and Schlag \cite{BGS1} proved Anderson localization for
\begin{equation}\label{7}
H=\lambda v(T^nx)+\Delta.
\end{equation}

In order to study the quantum kicked rotor equation
\begin{equation}\label{8}
i\frac{\partial\Psi(t,x)}{\partial t}=a\frac{\partial^{2}\Psi(t,x)}{\partial x^{2}}+ib\frac{\partial\Psi(t,x)}{\partial x}+V(t,x)\Psi(t,x),
\end{equation}
\begin{equation}\label{9}
V(t,x)=\kappa\left[\sum_{n\in\mathbb{Z}}\delta(t-n)\right]\cos2\pi x, 
\end{equation}
Bourgain \cite{B02} considered the lattice Schr\"odinger operator
\begin{equation}\label{10}
H(x)=v(T^mx)\delta_{mn}+\phi_{m-n}(T^mx)+\overline{\phi_{n-m}(T^nx)},
\end{equation}
where $v$ is a real, nonconstant, trigonometric polynomial, $\phi_k$ are trigonometric polynomials and $T$ is the skew shift on $\mathbb{T}^{2}$.
Using multi-scale method, Bourgain proved Anderson localization for the operator (\ref{10}).

In \cite{B02}, the quantum kicked rotor model is reduced to the monodromy operator
\begin{equation}\label{11}
W=e^{i\left(a\frac{d^{2}}{dx^{2}}+ib\frac{d}{dx}\right)}e^{i\kappa\cos2\pi x}.
\end{equation}
In this paper, we will study another reduction of the quantum kicked rotor model, which leads to an operator with simpler form.
More precisely, we consider
\begin{equation}\label{12}
H(x)=\tan\pi(T^mx)_{1}\delta_{mn}+\epsilon S_\phi,
\end{equation}
where $T$ is the skew shift on $\mathbb{T}^{2}$ and $(T^mx)_{1}$ refers to the first coordinate of $T^mx$. To make the operator (\ref{12}) well defined, we will always assume
\begin{equation}\label{13}
(T^mx)_{1}-\frac{1}{2} \notin \mathbb{Z},\quad \forall m \in \mathbb{Z}.
\end{equation}
The model (\ref{12}) is so-called the quantum kicked rotor model proposed by Fishman, Grempel and Prange, see Equation (3) in \cite{FGP}. As for this model, Bourgain \cite{B05} (p.120) remarked as follows:

\vskip0.2cm

{\it  ``This reduction is different from ours and leads to an operator with simpler form.
However, one has to deal with the singularity of the $tg$ function. It is likely that the
method explained in Chapters 14 and 15 also may be adapted to establish localization
results for (16.23)." }

\vskip0.2cm

In the present paper, we will fulfill Bourgain's idea as the above. More exactly, we will prove the following result.
\begin{thm}
Consider a lattice operator $H_{\omega}(x)$ associated to the skew shift $T=T_{\omega}$ of the form (\ref{12}).
Assume $\omega \in DC$ (diophantine condition)
\begin{equation}\label{14}
\|k\omega\|>c| k |^{-2},\quad \forall k\in\mathbb{Z}\setminus\{0\}
\end{equation}
and $\phi$  real analytic satisfying
\begin{equation}\label{15}
  |\hat{\phi}(n)|<e^{-\rho|n|},\quad \forall n \in \mathbb{Z}
\end{equation}
for some $\rho>0$. Fix $x_0\in\mathbb{T}^{2}$. Then for almost all $\omega \in DC$  and $\epsilon$ taken sufficiently small,
$H_{\omega}(x_0)$ satisfies Anderson localization.
\end{thm}

In the long range case here, the transfer matrix formalism is not applicable.
Our basic strategy is the same as that in \cite{B02}, but as mentioned above, the main difficulty is that
the potential $\tan$ is an unbounded function (i.e. of singularity) and the operator $H$ is unbounded.

In order to prove Anderson localization, we need Green's function estimates for
\begin{equation}\label{16}
G_{[0,N]}(x,E)=(R_{[0,N]}(H(x)-E)R_{[0,N]})^{-1},
\end{equation}
where $R_{\Lambda}$ is the restriction operator to $\Lambda\subset\mathbb{Z}$.
Note that
\begin{equation}\label{17}
R_{[0,N]}(H(x)-E)R_{[0,N]}=D(x)B(x),
\end{equation}
where
\begin{equation}\label{18}
D(x)=\mathrm{diag}\left(\frac{1}{\cos\pi x_{1}},\ldots,\frac{1}{\cos\pi (T^{N}x)_{1}}\right).
\end{equation}
Hence
\begin{equation}\label{19}
G_{[0,N]}(x,E)=B(x)^{-1}D(x)^{-1}.
\end{equation}
Since in $D(x)^{-1}$, the singularity $\frac{1}{\cos}$ vanishes, we only need Green's function estimates for $B(x)^{-1}$.
We need to point out that $B(x)$ is not self-adjoint. Fortunately, we find that multi-scale analysis still applies to this case.
Since the operator $H$ is unbounded and the energy $E$ is unbounded, we use the specific property of trigonometric functions to overcome
the difficulty of the unboundedness of the energy $E$.


We summarize the structure of this paper.
First, we will prove Green's function estimates in Section 2.
Then we recall some facts about semi-algebraic sets in Section 3 and give the proof of Anderson localization
in Section 4.

We will use the following notations. For positive numbers $a,b,a\lesssim b$ means $Ca\leq b$ for some constant $C>0$.
$a\ll b$ means $C$ is large. $a\sim b$ means $a\lesssim b$ and $b\lesssim a$. $N^{1-}$ means $N^{1-\epsilon}$ with some small $\epsilon>0$.
For $x\in\mathbb{R}$, $\| x\|=\inf\limits_{m\in\mathbb{Z}}|x-m|$,
for $x=(x_1, x_2)\in\mathbb{T}^{2}$,  $\| x\|=\| x_{1}\|+\| x_{2}\|$ .

\section{Green's function estimates}

In this section, we will prove the Green's function estimates using multi-scale analysis in \cite{B02}.

We need the following lemma.

\begin{lem}[Lemma 3.16 in \cite{B02}]\label{l2.1}

Let $A(x)=\{A_{mn}(x)\}_{1\leq m,n\leq N}$ be a matrix-valued function on $\mathbb{T}^{d}$ such that
\begin{equation}\label{b1}
\mbox{$A(x)$ is self-adjoint for $x\in\mathbb{T}^{d}$},
\end{equation}
\begin{equation}\label{b2}
\mbox{$A_{mn}(x)$ is a trigonometric polynomial of degree $<N^{C_{1}}$},
\end{equation}
\begin{equation}\label{b3}
|A_{mn}(x)|<C_2e^{-c_2|m-n|},
\end{equation}
where $c_2,C_1,C_2>0$ are constants.

Let $0<\delta<1$ be sufficiently small, $M=N^{\delta^{6}},\ L_0=N^{\frac{1}{100}\delta^{2}},\ 0<c_3<\frac{1}{10}c_2.$

Assume that for any interval $I\subset[1,N]$ of size $L_0$, except for $x$ in a set of measure $<e^{-L_0^{\delta^{3}}}$,
\begin{equation}\label{b4}
\|(R_IA(x)R_I)^{-1}\|<e^{L_0^{1-}},
\end{equation}
\begin{equation}\label{b5}
|(R_IA(x)R_I)^{-1}(m,n)|<e^{-c_3|m-n|},\quad m,n\in I,|m-n|>\frac{L_0}{10}.
\end{equation}

For fixed $x\in\mathbb{T}^{d}, n_0\in[1,N]$ is called a good site if $I_0=\left[n_0-\frac{M}{2},n_0+\frac{M}{2}\right]\subset[1,N]$ and
\begin{equation}\label{b6}
\|( R_{I_{0}}A(x)R_{I_{0}})^{-1}\|<e^{M^{1-}},
\end{equation}
\begin{equation}\label{b7}
| (R_{I_{0}}A(x)R_{I_{0}})^{-1}(m,n)|<e^{-c_3|m-n|},\quad m,n\in I_0,|m-n|>\frac{M}{10}.
\end{equation}

Denote $\Omega(x)\subset[1,N]$ the set of bad sites.
Assume that for any interval $J\subset[1,N]$ such that $|J|>N^{\frac{\delta}{5}}$, we have
\begin{equation}\label{b8}
|J\cap\Omega(x)|<|J|^{1-\delta}.
\end{equation}
Then
\begin{equation}\label{b9}
\| A(x)^{-1}\|<e^{N^{1-\frac{\delta}{C(d)}}},
\end{equation}
\begin{equation}\label{b10}
| A(x)^{-1}(m,n)|<e^{-c'_3|m-n|},\quad |m-n|>\frac{N}{10}
\end{equation}
except for $x$ in a set of measure $<e^{-\frac{N^{\delta^{2}}}{C(d)}}$, where $C(d)$ is a constant depending on $d$ and $c_{3}'>c_{3}-(\log N)^{-8}$.
\end{lem}

We also need the following ergodic property of skew shifts on $\mathbb{T}^{2}$.

\begin{lem}[Lemma 15.21 in \cite{B05}]\label{l2.2}

Assume $\omega \in DC$, $T=T_{\omega}$ is the skew shift on $\mathbb{T}^{2}$, $\epsilon>L^{-\frac{1}{10}}$. Then
\begin{equation*}
\#\left\{n=1,\ldots,L\Big\lvert\|T^nx-a\|<\epsilon\right\}<C\epsilon^{2}L.
\end{equation*}

\end{lem}

\begin{rem}\label{r2.3}

In the proof of Lemma \ref{l2.2}, we only need to assume
\begin{equation*}
  \|k\omega\|>c|k|^{-2},\quad \forall 0<|k|\leq L.
\end{equation*}

\end{rem}

By Lemma \ref{l2.1}, Lemma \ref{l2.2}, we can prove the Green's function estimates.

\begin{prop}\label{p2.4}

Let $T=T_{\omega}$ be the skew shift and
\begin{equation}\label{b11}
H_{mn}(x)=\tan\pi(T^mx)_{1}\delta_{mn}+\epsilon S_\phi.
\end{equation}
Assume $\phi$  real analytic satisfying
\begin{equation}\label{b12}
  |\hat{\phi}(n)|<e^{-\rho|n|},\quad \forall n \in \mathbb{Z}
\end{equation}
for some $\rho>0$ and $\omega$ satisfying
\begin{equation}\label{b13}
   \| k\omega\|>c|k|^{-2}, \quad \forall 0<|k|\leq N,
\end{equation}
$\epsilon$  is small.
Then for energy $E$,
\begin{equation}\label{b14}
   \| G_{[0,N]}(x,E)\| <e^{N^{1-}},
\end{equation}
\begin{equation}\label{b15}
  |G_{[0,N]}(x,E)(m,n)|<e^{-\frac{\rho}{100}|m-n|},\quad  0\leq m,n\leq N, |m-n|>\frac{N}{10}
\end{equation}
for $x\notin \Omega_{N}(E)$, where
 \begin{equation}\label{b16}
 {\rm mes} \Omega_{N}(E)<e^{-N^{\sigma}},\quad\sigma>0.
 \end{equation}

\end{prop}

\begin{proof}

Write
\begin{equation}\label{b17}
 H_{[0,N]}(x)-E=D_{[0,N]}(x)B_{[0,N]}(x),
\end{equation}
where
\begin{equation}\label{b18}
D_{mn}(x) =\frac{\sqrt{1+(\epsilon\hat{\phi}(0)-E)^{2}}}{\cos\pi(T^mx)_{1}} \delta_{mn},
\end{equation}
\begin{equation}\label{b19}
B_{mm}(x) =\frac{1}{\sqrt{1+(\epsilon\hat{\phi}(0)-E)^{2}}}\left[\sin\pi(T^mx)_{1}+(\epsilon\hat{\phi}(0)-E)\cos\pi(T^mx)_{1}\right],
\end{equation}
\begin{equation}\label{b20}
 B_{mn}(x) =\frac{ \epsilon\hat{\phi}(m-n)\cos\pi(T^mx)_{1}}{\sqrt{1+(\epsilon\hat{\phi}(0)-E)^{2}}}, \quad m\neq n.
\end{equation}

We will apply Lemma \ref{l2.1} to $B_{[0,N]}(x)$. Note that $B_{[0,N]}(x)$ is not self-adjoint.
However, in the proof of Lemma \ref{l2.1}, we don't need (\ref{b1}).
Since
\begin{equation}\label{b21}
T^m(x_1,x_2)=\left(x_1+mx_2+\frac{m(m-1)}{2}\omega,x_2+m\omega\right),
\end{equation}
$B_{mn}(x)$ is a trigonometric polynomial of degree $<|m|$.
(\ref{b3}) holds with $C_2=1, c_2=\rho$.

We need to prove
\begin{equation}\label{b22}
 {\rm mes}\left[x\in\mathbb{T}^{2}\Big\lvert|B_{[0,N]}(x)^{-1}(m,n)|>e^{N^{1-}-c_3|m-n|\chi_{|m-n|>\frac{N}{10}}} ,\exists 0\leq m,n\leq N\right]<e^{-N^{\delta^{3}}}
\end{equation}
for some $c_3>\frac{\rho}{100}, 0<\delta<1$.

By
\begin{equation*}
  |\sin\pi x+(\epsilon\hat{\phi}(0)-E)\cos\pi x|=\sqrt{1+(\epsilon\hat{\phi}(0)-E)^{2}}\left|\cos\pi(x-\alpha)\right|, \quad 0<\alpha<1,
\end{equation*}
using the fact
\begin{equation*}
   {\rm mes}\left[x\in[0,1]\Big\lvert |\cos\pi x|<\eta\right]<\eta, \quad  \forall 0<\eta<1,
\end{equation*}
we have
\begin{equation}\label{b23}
{\rm mes}\left[x\in[0,1]\Big\lvert \frac{1}{\sqrt{1+(\epsilon\hat{\phi}(0)-E)^{2}}}|\sin\pi x+(\epsilon\hat{\phi}(0)-E)\cos\pi x|<\epsilon_{0}\right]<\epsilon_{0}.
\end{equation}
Since $T$ is a measure-preserving transformation,
\begin{equation}\label{b24}
{\rm mes}\left[x\in\mathbb{T}^{2}\Big\lvert\frac{1}{\sqrt{1+(\epsilon\hat{\phi}(0)-E)^{2}}}|\sin\pi(T^mx)_{1}+(\epsilon\hat{\phi}(0)-E)\cos\pi(T^mx)_{1}|<\epsilon_{0}\right]<\epsilon_{0}.
\end{equation}
Hence
\begin{equation}\label{b25}
{\rm mes}\left[x\in\mathbb{T}^{2}\Big\lvert\min_{0\leq m\leq N_0} |B_{mm}(x)|<\epsilon_{0}\right]<N_0\epsilon_{0}.
\end{equation}

If $\min\limits_{0\leq m\leq N_0} |B_{mm}(x)|>\epsilon_{0}>\epsilon$, take $\epsilon_{0}=e^{-N_0^{\frac{1}{2}}},\epsilon=e^{-N_0}$,
by Neumann expansion and (\ref{b25}), we have
\begin{equation}\label{b26}
 |B_{[0,N_{0}]}(x)^{-1}(m,n)|<e^{N_0^{\frac{1}{2}}-\frac{\rho}{2}|m-n|}, \quad m,n\in[0,N_0]
\end{equation}
except for $x$ in a set of measure $<e^{-cN_0^{\frac{1}{2}}}$.
So, (\ref{b22}) holds for an initial scale $N_0$.

Assume (\ref{b22}) holds up to scale $L_0=N^{\frac{1}{100}\delta^{2}}$, since
\begin{equation}\label{b27}
B_{m+1,n+1}(x)=B_{mn}(Tx),
\end{equation}
(\ref{b4}),(\ref{b5}) will hold for $x$ outside a set of measure at most $e^{-L_0^{\delta^{3}}}$.
Denote $\Omega(x)\subset[0,N]$ the set of bad sites with respect to scale $M=N^{\delta^{6}}$. $n_0\notin\Omega(x)$ means
\begin{equation}\label{b28}
|B_{[0,M]}(T^{n_{0}-\frac{M}{2}}x)^{-1}(m,n)|=\left|B_{[n_{0}-\frac{M}{2},n_{0}+\frac{M}{2}]}(x)^{-1}\left(m+n_{0}-\frac{M}{2},n+n_{0}-\frac{M}{2}\right)\right|
\end{equation}
\begin{equation*}
<e^{M^{1-}-c_3|m-n|\chi_{|m-n|>\frac{M}{10}}}, \quad m,n\in[0,M].
\end{equation*}

From the inductive hypothesis, we have
\begin{align}\label{b29}
  |B_{[0,M]}(x)^{-1}(m,n)|<e^{M^{1-}-c_3|m-n|\chi_{|m-n|>\frac{M}{10}}}, \quad m,n\in[0,M]
\end{align}
for $x\notin\Omega_{0}, {\rm mes}\Omega_{0}<e^{-M^{\delta^{3}}}$.
By (\ref{b28}), (\ref{b29}), Lemma \ref{l2.1}, we only need to show that for any $x\in\mathbb{T}^{2},\ N^{\frac{\delta}{5}}<L<N$,
\begin{align}\label{b30}
  \#\{1\leq n\leq L|T^nx\in\Omega_{0}\}<L^{1-\delta}.
\end{align}

Expressing (\ref{b29}) as a ratio of determinants and replacing $\cos, \sin$ by truncated power series,
$\Omega_{0}$ may be viewed as a semi-algebraic set of degree at most $M^{6}$.
(For properties of semi-algebraic sets, see Section 3.)
If $r>e^{-\frac{1}{2}M^{\delta^{3}}}$, by Proposition \ref{p3.2},  $\Omega_{0}$
may be covered by at most $M^{C}\left(\frac{1}{r}\right) r$-balls.
Choosing $r=L^{-\frac{1}{20}}>N^{-1}>e^{-\frac{1}{2}M^{\delta^{3}}}$, using Lemma \ref{l2.2}, Remark \ref{r2.3}, we have
\begin{equation*}
  \#\{1\leq n\leq L|T^nx\in\Omega_{0}\}<M^{C}\left(\frac{1}{r}\right)r^{2}L<L^{C\delta^{5}+1-\frac{1}{20}}<L^{1-\delta}.
\end{equation*}
This proves (\ref{b30}) and (\ref{b22}).

By (\ref{b17}),
\begin{align}\label{b31}
G_{[0,N]}(x,E)= (H_{[0,N]}(x)-E)^{-1}=B_{[0,N]}(x)^{-1}D_{[0,N]}(x)^{-1},
\end{align}
hence
\begin{align}\label{b32}
G_{[0,N]}(x,E)(m,n)=\frac{\cos\pi(T^nx)_{1}}{\sqrt{1+(\epsilon\hat{\phi}(0)-E)^{2}}}B_{[0,N]}(x)^{-1}(m,n), \quad m,n\in[0,N].
\end{align}

By (\ref{b31}), (\ref{b32}),
\begin{align}\label{b33}
\|G_{[0,N]}(x,E)\|\leq\|B_{[0,N]}(x)^{-1}\|,
\end{align}
\begin{align}\label{b34}
|G_{[0,N]}(x,E)(m,n)|\leq|B_{[0,N]}(x)^{-1}(m,n)|, \quad m,n\in[0,N].
\end{align}
Proposition \ref{p2.4} follows from (\ref{b22}), (\ref{b33}), (\ref{b34}).
\end{proof}

\section{Semi-algebraic sets}

We recall some basic facts of semi-algebraic sets in this section, which is needed in Section 4. Let $\mathcal{P}=\{P_1,\ldots,P_s\}\subset\mathbb{R}[X_1,\ldots,X_n]$
be a family of real polynomials whose degrees are bounded by $d$.
A semi-algebraic set is given by
\begin{equation}\label{ss}
S=\bigcup_{j}\bigcap_{l\in L_{j}}\left\{\mathbb{R}^{n}\Big\lvert P_ls_{jl}0\right\},
\end{equation}
where $L_{j}\subset\{1,\ldots,s\},s_{jl}\in\{\leq,\geq,=\}$ are arbitrary.
We say that $S$ has degree at most $sd$ and its degree is the $\inf$ of $sd$ over all representations as in (\ref{ss}).

We need the following quantitative version of the Tarski-Seidenberg principle.

\begin{prop}[\cite{BPR}]\label{p3.1}
Let $S\subset\mathbb{R}^{n}$ be a semi-algebraic set of degree $B$, then any projection of $S$ is semi-algebraic of degree at most $B^{C}, C=C(n)$.
\end{prop}

We also need the following fact.

\begin{prop}[Corollary 9.6 in \cite{B05}]\label{p3.2}
 Let $S\subset[0,1]^{n}$ be semi-algebraic of degree $B$.
 Let $\epsilon>0, \ {\rm mes}_nS<\epsilon^{n}$. Then $S$ may be covered by at most $B^{C}(\frac{1}{\epsilon})^{n-1} \epsilon$-balls.
\end{prop}

Finally, we will use the following lemma.

\begin{lem}[Lemma 15.26 in \cite{B05}]\label{l3.3}

 Let $S\subset\mathbb{T}^{3}$ be a semi-algebraic set of degree $B$ such that
 \begin{equation*}
{\rm mes}S<e^{-B^{\sigma}},\quad \sigma>0.
 \end{equation*}
Let $M$ be an integer satisfying
\begin{equation*}
\log\log M\ll\log B\ll\log M.
\end{equation*}
Then for any fixed $x_{0}\in \mathbb{T}^{2}$,
\begin{equation*}
{\rm mes}[\omega\in\mathbb{T}\lvert(\omega,T^{j}_{\omega}x_{0})\in S,\ \exists j\sim M]<M^{-c},\quad c>0
\end{equation*}
where $T_{\omega}$ is the skew shift with frequency $\omega$.

\end{lem}

\section{Proof of Anderson localization}

In this section, we give the proof of Anderson localization as in \cite{BG}.

By application of the resolvent identity, we have the following

\begin{lem}\label{l4.1}

Let $I\subset\mathbb{Z}$ be an interval of size $N$ and $\{I_{\alpha}\}$ subintervals of size $M\ll N$, $N=e^{(\log M)^{2}}$.
Assume $\forall k\in I$, there is some $\alpha$ such that
\begin{equation}\label{c1}
\left[k-\frac{M}{4},k+\frac{M}{4}\right]\cap I\subset I_\alpha
\end{equation}
and $\forall \alpha$,
\begin{equation}\label{c2}
\|G_{I_{\alpha}}\|<e^{M^{1-}},
|G_{I_{\alpha}}(n_1,n_2)|<e^{-\frac{\rho}{100}|n_1-n_2|},  \ n_1,n_2\in I_{\alpha},|n_1-n_2|>\frac{M}{10}.
\end{equation}
Then
\begin{equation}\label{c3}
|G_{I}(n_1,n_2)|<e^{M}, \quad  n_1,n_2\in I ,
\end{equation}
\begin{equation}\label{c4}
|G_{I}(n_1,n_2)|<e^{-\frac{\rho}{200}|n_1-n_2|}, \quad  n_1,n_2\in I ,|n_1-n_2|>\frac{N}{10}.
\end{equation}

\end{lem}

\begin{proof}

For $m,n\in I$, there is some $\alpha$ such that
\begin{equation}\label{c5}
\left[m-\frac{M}{4},m+\frac{M}{4}\right]\cap I\subset I_\alpha.
\end{equation}
By resolvent identity,
\begin{equation}\label{c6}
|G_{I}(m,n)|\leq e^{M^{1-}}+\sum_{n_1\in I_{\alpha},n_2\notin I_{\alpha}}|G_{I_{\alpha}}(m,n_1)|e^{-\rho|n_1-n_2|}|G_{I}(n_{2},n)|.
\end{equation}
If $|m-n_1|\leq\frac{M}{8}$, then $|n_1-n_2|\geq\frac{M}{8}$, hence
\begin{equation}\label{c7}
\sum_{|m-n_1|\leq\frac{M}{8},m_2\notin I_{\alpha}}|G_{I_{\alpha}}(m,n_1)|e^{-\rho|n_1-n_2|}\leq M e^{M^{1-}}e^{-\rho\frac{M}{8}}<\frac{1}{4}.
\end{equation}
If $|m-n_1|>\frac{M}{8}$, then $|G_{I_{\alpha}}(m,n_1)|<e^{-\frac{\rho}{100}|m-n_1|}$, hence
\begin{equation}\label{c8}
\sum_{|m-n_1|>\frac{M}{8},m_2\notin I_{\alpha}}|G_{I_{\alpha}}(m,n_1)|e^{-\rho|n_1-n_2|}< e^{-\frac{\rho}{1000}M}<\frac{1}{4}.
\end{equation}
By (\ref{c6}), (\ref{c7}), (\ref{c8}),
\begin{equation}\label{c9}
\max_{m,n\in I}|G_{I}(m,n)|<e^{M^{1-}}+\frac{1}{2}\max_{m,n\in I}|G_{I}(m,n)|.
\end{equation}
(\ref{c3}) follows from (\ref{c9}).

Take $m,n\in I, |m-n|>\frac{N}{10}$, assume (\ref{c5}), by resolvent identity,
\begin{align}\label{c10}
\nonumber|G_{I}(m,n)|&\leq \sum_{n_0\in I_{\alpha},n_1\notin I_{\alpha}}|G_{I_{\alpha}}(m,n_0)|e^{-\rho|n_0-n_1|}|G_{I}(n_1,n)|\\
\nonumber&\leq M \sum_{|m-n_1|>\frac{M}{4}}e^{-\frac{\rho}{100}|m-n_1|}|G_{I}(n_1,n)|\\
& \leq M^{t}\sum_{|m-n_1|>\frac{M}{4},\ldots,|n_{t-1}-n_t|>\frac{M}{4}}e^{-\frac{\rho}{100}(|m-n_1|+\cdots+|n_{t-1}-n_t|)}|G_{I}(n_t,n)|
\end{align}
where $t\leq 10\frac{N}{M}$.

If $|n-n_t|\leq M$, then by (\ref{c3}), (\ref{c10}),
\begin{equation}\label{c11}
 |G_{I}(m,n)|  \leq M^{t}N^{t}e^{M-\frac{\rho}{100}|m-n_{t}|}
 \leq e^{20\frac{N}{M}\log N+2M-\frac{\rho}{100}|m-n|}< e^{-\frac{\rho}{200}|m-n|}.
\end{equation}

If $t= 10\frac{N}{M}$, then by (\ref{c3}), (\ref{c10}),
\begin{equation}\label{c12}
 |G_{I}(m,n)|  \leq M^{t}N^{t}e^{-\frac{\rho}{100}\frac{10N}{M}\frac{M}{4}+M}
\leq e^{20\frac{N}{M}\log N+M-\frac{\rho}{100}2N}<e^{-\frac{\rho}{100}|m-n|}.
\end{equation}

(\ref{c4}) follows from (\ref{c11}), (\ref{c12}). This proves Lemma \ref{l4.1}.
\end{proof}

Now we can prove the main result.

\begin{thm}\label{t4.2}

Consider a lattice operator $H_{\omega}(x)$ associated to the skew shift $T=T_{\omega}$ of the form
\begin{equation}\label{c13}
H_{\omega}(x)=\tan\pi(T^mx)_{1}\delta_{mn}+\epsilon S_\phi.
\end{equation}
Assume $\omega \in DC$ (diophantine condition)
\begin{equation}\label{c14}
\|k\omega\|>c| k |^{-2},\quad \forall k\in\mathbb{Z}\setminus\{0\}
\end{equation}
and $\phi$  real analytic satisfying
\begin{equation}\label{c15}
  |\hat{\phi}(n)|<e^{-\rho|n|},\quad \forall n \in \mathbb{Z}
\end{equation}
for some $\rho>0$. Fix $x_0\in\mathbb{T}^{2}$. Then for almost all $\omega \in DC$  and $\epsilon$ taken sufficiently small,
$H_{\omega}(x_0)$ satisfies Anderson localization.

\end{thm}

\begin{proof}

By Shnol's theorem \cite{H},
to establish Anderson localization, it suffices to show that if $\xi=(\xi_n)_{n\in\mathbb{Z}},E\in\mathbb{R}$ satisfy
\begin{equation}\label{c16}
\xi_0=1, |\xi_n|<C|n|,\quad |n|\rightarrow\infty,
\end{equation}
\begin{equation}\label{c17}
H(x_0)\xi=E\xi,
\end{equation}
then
\begin{equation}\label{c18}
|\xi_n|<e^{-c|n|},\quad |n|\rightarrow\infty.
\end{equation}

Denote $\Omega=\Omega(E)\subset\mathbb{T}^{2}$ the set of $x$ such that
\begin{equation}\label{c19}
  |G_{[-N,N]}(x,E)(m,n)|<e^{N^{1-}-\frac{\rho}{100}|m-n|\chi_{|m-n|>\frac{N}{10}}}
\end{equation}
fails for some $|m|,|n|\leq N$.
Let $N_1=N^{C_{1}}$, $C_{1}$ is a sufficiently large constant.
Then by Proposition \ref{p2.4},
\begin{equation}\label{c20}
 {\rm mes} \Omega(E)<e^{-N^{\sigma}},
\end{equation}
\begin{equation}\label{c21}
\#\{|j|\leq N_1|T^{j}x_{0}\in\Omega\}<N_1^{1-\delta}.
\end{equation}
So, we may find an interval $I\subset[0,N_1]$ of size $N$ such that
\begin{equation}\label{c22}
T^{j}x_{0}\notin\Omega, \quad\forall j\in I \cup(-I).
\end{equation}
Hence
\begin{equation}\label{c23}
|G_{[j-N, j+N]}(x_{0},E)(m,n)|<e^{N^{1-}-\frac{\rho}{100}|m-n|\chi_{|m-n|>\frac{N}{10}}},\quad m,n\in [j-N, j+N].
\end{equation}
By (\ref{c16}), (\ref{c17}), (\ref{c23}), we have
\begin{equation}\label{c24}
|\xi_{j}|\leq C\sum_{n_{1}\in [j-N, j+N],n_{2}\notin [j-N, j+N]}e^{N^{1-}-\frac{\rho}{100}|j-n_{1}|\chi_{|j-n_{1}|>\frac{N}{10}}}e^{-\rho|n_{1}-n_{2}|}|n_{2}|<e^{-\frac{\rho}{200}N}.
\end{equation}
Denoting $j_0$ the center of $I$, we have
\begin{equation}\label{c25}
1=\xi_{0}\leq\|G_{[-j_0,j_0]}(x_0,E)\|\|R_{[-j_0,j_0]}H(x_0)R_{\mathbb{Z}\setminus[-j_0,j_0]}\xi\|.
\end{equation}

By (\ref{c16}), (\ref{c24}), we have for $|n|\leq j_0$,
\begin{equation}\label{c26}
|(R_{[-j_0,j_0]}H(x_0)R_{\mathbb{Z}\setminus[-j_0,j_0]}\xi)_{n}|\leq\sum_{|n_{1}|>j_0}e^{-\rho|n-n_{1}|}|\xi_{n_1}|
\end{equation}
\begin{equation*}
  \leq\sum_{j_0<|n_1|\leq j_0+\frac{N}{2}}e^{-\rho|n-n_1|}e^{-\frac{\rho}{200}N}+C\sum_{|n_1|>j_0+\frac{N}{2}}e^{-\rho|n-n_1|}|n_1|
<e^{-\frac{\rho}{400}N}.
\end{equation*}

By (\ref{c25}), (\ref{c26}),
\begin{equation}\label{c27}
\|G_{[-j_0,j_0]}(x_0,E)\|>e^{\frac{\rho}{500}N},
\end{equation}
hence
\begin{equation}\label{c28}
{\rm dist}(E, {\rm spec} H_{[-j_0,j_0]}(x_0))<e^{-\frac{\rho}{500}N}.
\end{equation}

Denote
\begin{equation}\label{c29}
\mathcal{E}_{\omega}=\bigcup_{|j|\leq N_1}{\rm spec} H_{[-j_0,j_0]}(x_0).
\end{equation}
It follows from (\ref{c28}) that if $x\notin\bigcup\limits_{E'\in\mathcal{E}_{\omega}}\Omega(E')$,
then
\begin{equation}\label{c30}
  |G_{[-N,N]}(x,E)(m,n)|<e^{N^{1-}-\frac{\rho}{100}|m-n|\chi_{|m-n|>\frac{N}{10}}} ,\quad |m|,|n|\leq N.
\end{equation}

Consider the set $S\subset\mathbb{T}^{3}\times\mathbb{R}$ of $(\omega,x,E')$, where
\begin{equation}\label{c31}
 \| k\omega\|>c|k|^{-2},\quad \forall 0<|k|\leq N,
\end{equation}
\begin{equation}\label{c32}
x\in\Omega(E'),
\end{equation}
\begin{equation}\label{c33}
E'\in\mathcal{E}_{\omega}.
\end{equation}

By Proposition \ref{p3.1},
\begin{equation}\label{c34}
\mbox{${\rm Proj}_{\mathbb{T}^{3}}S$ is a semi-algebraic set of degree $<N^C$},
\end{equation}
and by (\ref{c20}),
\begin{equation}\label{c35}
{\rm mes}({\rm Proj}_{\mathbb{T}^{3}}S)<e^{-\frac{1}{2}N^{\sigma}}.
\end{equation}

Let $N_2=e^{(\log N)^{2}}$,
\begin{equation}\label{c36}
\mathcal{R}_{N}=\{\omega\in\mathbb{T}\mid(\omega,T^{j}x_{0})\in {\rm Proj}_{\mathbb{T}^{3}}S, \ \exists|j|\sim N_2\}.
\end{equation}
By (\ref{c34}), (\ref{c35}), (\ref{c36}), using Lemma \ref{l3.3}, ${\rm mes}\mathcal{R}_{N}<N_2^{-c},c>0$.
Let
\begin{equation}\label{c37}
\mathcal{R}=\bigcap_{N_0\geq 1}\bigcup_{N\geq N_0}\mathcal{R}_{N},
\end{equation}
then ${\rm mes}\mathcal{R}=0$.
We restrict $\omega\notin\mathcal{R}$.

If $\omega\notin\mathcal{R}_N$, we have for all $|j|\sim N_2, \ (\omega,T^{j}x_{0})\notin {\rm Proj}_{\mathbb{T}^{3}}S$,
by (\ref{c30}),
\begin{equation}\label{c38}
  |G_{[j-N,j+N]}(x_0,E)(m,n)|<e^{N^{1-}-\frac{\rho}{100}|m-n|\chi_{|m-n|>\frac{N}{10}}},\quad m,n\in [j-N, j+N].
\end{equation}

Let $\Lambda=\bigcup\limits_{\frac{1}{4}N_2<j<4N_2}[j-N,j+N]\supset\left[\frac{1}{4}N_2,4N_2\right]$, by Lemma \ref{l4.1}, we deduce from (\ref{c38}) that
\begin{equation}\label{c39}
  |G_{\Lambda}(x_0,E)(m,n)|<e^{-\frac{\rho}{200}|m-n|}, \quad |m-n|>\frac{N_2}{10},
\end{equation}
and therefore
\begin{equation}\label{c40}
 |\xi_j|<e^{-\frac{\rho}{4000}|j|},\quad  \frac{1}{2}N_2\leq j\leq N_2 .
\end{equation}

Since $\omega\notin\mathcal{R}$, by (\ref{c37}), there is some $N_0>0$, such that for all $N\geq N_0,\omega\notin\mathcal{R}_N$.
So, (\ref{c40}) holds for $j\in\bigcup\limits_{N\geq N_0}[\frac{1}{2}e^{(\log N)^{2}},e^{(\log N)^{2}}]=[\frac{1}{2}e^{(\log N_0)^{2}},\infty)$.
This proves (\ref{c18}) for $j>0$, similarly for $j<0$. Hence
Theorem \ref{t4.2} follows.
\end{proof}

\subsection*{Acknowledgment}
This paper was supported by  National Natural
Science Foundation of China (No. 11790272 and No. 11771093).

\end{document}